\documentclass{article}

\PassOptionsToPackage{numbers,compress}{natbib}
\usepackage[main,final]{neurips_2026}

\usepackage[T1]{fontenc}
\usepackage[utf8]{inputenc}
\usepackage{microtype}
\usepackage{amsmath,amssymb,amsthm,mathtools}
\usepackage{aliascnt}
\usepackage{enumitem}
\usepackage{booktabs}
\usepackage{xcolor}
\usepackage{url}
\usepackage[pdfusetitle,hidelinks,bookmarksnumbered=true]{hyperref}
\hypersetup{pdfauthor={Tingyi Lin, Jiazhuo Li, Ruoran Lai}}
\usepackage{orcidlink}
\usepackage[nameinlink,noabbrev]{cleveref}

\newtheorem{theorem}{Theorem}[section]
\crefname{theorem}{Theorem}{Theorems}
\Crefname{theorem}{Theorem}{Theorems}
\newaliascnt{lemma}{theorem}
\newtheorem{lemma}[lemma]{Lemma}
\aliascntresetthe{lemma}
\crefname{lemma}{Lemma}{Lemmas}
\Crefname{lemma}{Lemma}{Lemmas}
\newaliascnt{proposition}{theorem}
\newtheorem{proposition}[proposition]{Proposition}
\aliascntresetthe{proposition}
\crefname{proposition}{Proposition}{Propositions}
\Crefname{proposition}{Proposition}{Propositions}
\newaliascnt{corollary}{theorem}
\newtheorem{corollary}[corollary]{Corollary}
\aliascntresetthe{corollary}
\crefname{corollary}{Corollary}{Corollaries}
\Crefname{corollary}{Corollary}{Corollaries}
\theoremstyle{definition}
\newaliascnt{definition}{theorem}
\newtheorem{definition}[definition]{Definition}
\aliascntresetthe{definition}
\crefname{definition}{Definition}{Definitions}
\Crefname{definition}{Definition}{Definitions}
\newaliascnt{example}{theorem}
\newtheorem{example}[example]{Example}
\aliascntresetthe{example}
\crefname{example}{Example}{Examples}
\Crefname{example}{Example}{Examples}
\theoremstyle{remark}
\newaliascnt{remark}{theorem}
\newtheorem{remark}[remark]{Remark}
\aliascntresetthe{remark}
\crefname{remark}{Remark}{Remarks}
\Crefname{remark}{Remark}{Remarks}
\crefname{table}{Table}{Tables}
\Crefname{table}{Table}{Tables}
\crefname{section}{Section}{Sections}
\Crefname{section}{Section}{Sections}
\crefname{subsection}{Section}{Sections}
\Crefname{subsection}{Section}{Sections}
\crefname{appendix}{Appendix}{Appendices}
\Crefname{appendix}{Appendix}{Appendices}
\crefname{figure}{Figure}{Figures}
\Crefname{figure}{Figure}{Figures}
\theoremstyle{plain}
\newtheorem{maintheorem}{Theorem}

\newcommand{\R}{\mathbb{R}}
\newcommand{\E}{\mathbb{E}}
\newcommand{\Aset}{\mathcal{A}}
\newcommand{\Sig}{\Sigma}

\title{How Much Must a Private Mempool Hide? Exact Leakage Thresholds for Sandwich Attacks}
\author{%
  Tingyi Lin$^{1\dagger}$\,\orcidlink{0009-0001-9576-0912}, Jiazhuo Li$^{2}$, Ruoran Lai$^{3}$\\
  $^{1}$Tropic Research Institute\quad $^{2}$University of Michigan\quad $^{3}$Sun Yat-sen University
}
\date{}

\begin{document}
\maketitle
{\renewcommand{\thefootnote}{\fnsymbol{footnote}}%
\footnotetext[2]{Correspondence to: tingyi3@illinois.edu}}

\begin{abstract}
Private and encrypted mempools hide pending transactions to stop sandwich
attacks and other forms of maximal extractable value (MEV), but what they hide
is rarely everything: a transaction's pair, direction, and a coarse range for
its size can still leak. How much leakage makes sandwiching pay? We answer
exactly for a fee-free constant-product automated market maker, the pricing rule
behind Uniswap v2. Traders observe an interval
containing the victim's size and bid in a first-price auction for the right to
sandwich it, and the winning front-run must keep the victim's trade executable
at every size in the interval. The answer turns on the smallest size consistent
with the leak. It alone determines the feasible front-runs, the largest feasible
front-run is optimal for pointwise, expected, and worst-case profit alike, and
the guaranteed profit has a closed form. When execution is costly, a privacy
layer that wants to rule out sandwiches profitable at every consistent size may
therefore reveal anything about the size except a lower bound above an explicit
threshold; the upper end of the range is irrelevant. With two or more symmetric
traders, every pure-strategy perfect Bayesian equilibrium of the auction hands
the entire expected net rent to the auctioneer. If the direction is hidden too, no
non-contingent first leg front-runs both possible directions, while post-trade
arbitrage can survive even perfect pre-trade hiding.
\end{abstract}

\begin{quote}
    \itshape
    ``A fool uttereth all his mind: but a wise man keepeth it in till afterwards.''
    \par
    \hfill -- Proverbs 29:11 (King James Version)
\end{quote}
\section{Introduction}

A hidden order has size $q>0$, but outside agents observe only a coarse public
signal. When does such leakage suffice to support profitable predatory behavior
in equilibrium? We answer this question exactly in the trading game induced by a
constant-product automated market maker (AMM), a smart contract that exchanges
two tokens at prices set by its reserves. The motivating examples come from
hidden-order execution environments, but only the induced information structure
enters the theorems.

If an adversarial trader knows that a latent order intends to buy token $Y$ using
token $X$, then a front-run purchase of $Y$ raises that order's execution price
and creates the standard sandwich pattern: trade in the latent order's
direction, let the latent order move the pool further, and then unwind in the
opposite direction. Under full observability, the trader may know the hidden
order exactly. Our interest is the intermediate regime in which the trader sees
only coarse information. Sandwiching remains common on Ethereum
\cite{TorresCaminoState21,QinZhouGervais22} and is not confined to simple
Uniswap v2 paths: McLaughlin, Kruegel, and Vigna identify
63{,}257 sandwich attacks among apparent arbitrages, including one whose
manipulation and unwind route through both Uniswap v2 and v3
\cite{McLaughlinKruegelVigna23}, and Bai et al.\ detect 60{,}946 sandwich
events, multi-token routes included, in 210{,}000 Ethereum blocks
\cite{BaiLiJiangDu25}.

A round of our game has four steps. A victim submits a swap whose size is
hidden, with a slippage guard that cancels the swap if its output falls too far
below the honest output. A public signal then reveals an interval that contains
the hidden size. Traders bid for the single right to surround the swap, with a
front-run immediately before it and an unwind immediately after. The winner
fixes its front-run size without learning the exact size, and the front-run must
leave the swap executable at every size consistent with the signal.

Formally, a leakage map sends the hidden size $q$ to a public signal $s$ whose
interval $I_s\ni q$ is either $[\ell_s,u_s]$ with lower endpoint $\ell_s>0$ or
$(0,u_s]$. The victim's direction and slippage tolerance $\tau\in(0,1)$ are
public. Before bidding, every trader and the auctioneer observe the pool
reserves $(X,Y)$, the trading pair, the direction, $\tau$, and $s$, but not
$q$, and they share a posterior $\mu_s$ on $I_s$ derived from a common prior.
Bids may use only \emph{robustly admissible} front-runs, which keep the victim's
trade valid for every $q\in I_s$.

Our first result characterizes the feasible front-runs and the profit they
guarantee.

\begin{maintheorem}[Admissibility, value, and threshold]
Fix a leaked interval $[\ell_s,u_s]$ with $\ell_s>0$.
\begin{enumerate}[label=(\roman*),leftmargin=2em]
\item The robustly admissible front-runs form an interval
$[0,a_{\max}(\ell_s,\tau)]$ whose endpoint is an explicit function of $X$,
$\tau$, and $\ell_s$ alone.
\item The largest admissible front-run is the unique maximizer of pointwise,
expected, and worst-case profit, and its worst-case gross profit is
\[
\frac{\tau \ell_s (X+\ell_s)}{X+\tau \ell_s}.
\]
\item With execution cost $c\ge 0$, a robustly admissible bundle with positive
net payoff at every consistent size exists if and only if $\ell_s$ exceeds an
explicit threshold $\lambda_c$ that depends only on $X$, $\tau$, and $c$.
\end{enumerate}
\end{maintheorem}

Theorem~A collects \Cref{thm:admissibility,thm:optimal,thm:threshold}. Only
the lower endpoint matters because a front-run hurts small victims most: the
victim's output after the front-run, as a fraction of its honest output, rises
with its size (\Cref{lem:ratio}). For a half-open interval $(0,u_s]$, parts
(i) and (ii) hold with $\ell_s=0$ for pointwise and expected profit, and with
$c>0$ no such bundle exists (\Cref{cor:perfect} with $u_s$ in place of $Q$).
\Cref{fig:threshold} plots the threshold and the effect of leaking the leading
bits of the size.

\begin{figure}[t]
\centering
\includegraphics[width=\linewidth]{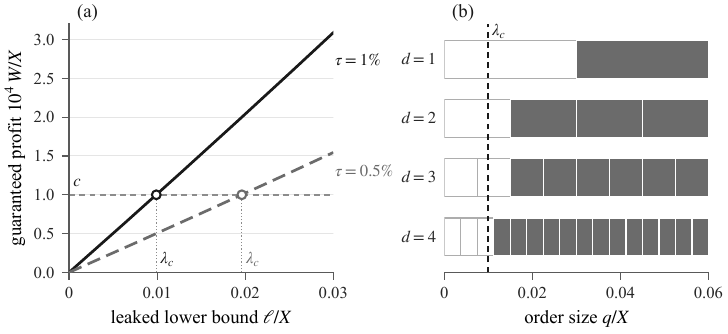}
\caption{(a) Guaranteed gross profit $W$ of the optimal robust sandwich against
the leaked lower bound $\ell$, for slippage tolerances $\tau=1\%$ and
$\tau=0.5\%$, with execution cost $c=10^{-4}X$ (dashed). A robust sandwich with
positive net payoff at every consistent size exists exactly when
$\ell>\lambda_c$ (\Cref{thm:threshold}); to first order $\lambda_c\approx
c/\tau$. (b) Leakage of the first $d$ bits of the size, for sizes up to
$Q=0.06X$ with $\tau=1\%$ and $c=10^{-4}X$. Shaded bins admit such a sandwich
(\Cref{prop:bitprefix}); the unshaded range shrinks toward
$\lambda_c\approx 0.0099X$ as $d$ grows and always contains $(0,\lambda_c]$.}
\label{fig:threshold}
\end{figure}

Our second result prices the execution right.

\begin{maintheorem}[Execution auction]
Let $m\ge 2$ symmetric risk-neutral traders bid for the slot, and let $R_s$ be
the posterior expected gross profit of the largest admissible front-run minus
the execution cost $c$. If $R_s>0$, every pure-strategy perfect Bayesian
equilibrium awards the slot to that front-run at the payment $R_s$, so the
auctioneer collects the entire expected net rent; if $R_s<0$, no equilibrium
places a sandwich.
\end{maintheorem}

Theorem~B is \Cref{thm:pbe}; the posterior enters only through $R_s$. With a
positive execution cost, a signal that reveals no positive lower bound supports
no guaranteed-profit sandwich, while post-trade arbitrage with gross profit up
to $q^2/(X+q)$ can survive (\Cref{cor:perfect,prop:boundary}).

Sandwich attacks on AMMs have been studied analytically under full
observability, most directly in Heimbach and Wattenhofer's sandwich game
\cite{HeimbachWattenhofer22}. Under partial observability the lower endpoint of
the leaked interval is the sufficient statistic: it alone determines robust
admissibility, the optimal robust attack, and the exact profitability
threshold, a leakage-aware analogue of the minimum profitable victim input in
public-order-flow analyses \cite{ZhouEtAl21,A2MM21}.

The baseline hides the size while the pair, the direction, and the slippage
tolerance are public, and the public direction is a substantive assumption. If
the direction may be either way and the first leg must be chosen before it is
revealed, no non-null first leg front-runs both directions
(\Cref{prop:direction}); a hidden slippage tolerance instead reduces to its
public lower bound (\Cref{prop:hidden-tau}). Hiding the direction thus blocks
the non-contingent robust sandwich studied here, and when the direction leaks
while the size is only interval-revealed, $\ell_s$ determines the exact
threshold. The closed forms use the algebra of the fee-free constant-product
invariant; \Cref{prop:fee} bounds the effect of swap fees in one direction, and
\Cref{rem:cfmm} discusses other pricing rules, for which we do not claim them.

\section{Related Work}
\label{sec:related}

Analytical work on sandwiching and transaction reordering under full or nearly
full observability provides the closest starting point. Daian et al.'s
\emph{Flash Boys 2.0} articulated the broader reordering perspective
\cite{DaianEtAl19}. Zhou et al.\ analyze sandwich attacks under public order
flow \cite{ZhouEtAl21}, and Heimbach and Wattenhofer study slippage-tolerance
choice in a sandwich game \cite{HeimbachWattenhofer22}. Park shows that under
transparent execution every liquidity-invariance pricing rule admits sandwich
attacks \cite{Park23}, and Kulkarni, Diamandis, and Chitra analyze routing and
reordering MEV in CFMMs, showing that the price of anarchy of routing is
constant when the impact of a sandwich attack is localized
\cite{KulkarniDiamandisChitra22}. In our notation, the
fully observed regime corresponds to point leakage $I=[q,q]$. Our results show
that under interval leakage the correct analogue of the minimum profitable
victim input intuition from A2MM \cite{A2MM21} is the leaked lower endpoint
$\ell$.

A second adjacent line studies market design and mechanism design for strategic
trading environments. Chan, Wu, and Shi analyze AMM mechanism design
\cite{ChanWuShi24}. Budish, Cramton, and Shim show how batch processing
destroys speed-based rents in another market-design setting
\cite{BudishCramtonShim15}. Wadhwa et al.\ study order-policy enforcement under
rationality and prove impossibility results in a fully rational setting
\cite{WadhwaEtAl23}. Kelkar et al.\ introduce order-fairness for Byzantine
consensus \cite{KelkarEtAl20}, Li et al.\ revisit transaction-fairness
definitions for blockchains \cite{LiEtAl23}, and Ferreira and Parkes design
verifiable sequencing rules that limit what a block producer gains by ordering
trades around a user \cite{FerreiraParkes23}. PROF studies protected order flow
in a profit-seeking PBS environment \cite{BabelEtAl24}. Transaction fee
mechanism design studies how block space should be allocated and priced when
block producers are strategic \cite{Roughgarden24,ChungShi23}, including
producers who extract MEV themselves \cite{BahraniGarimidiRoughgarden25}; the auctioneer of our execution-rights
auction sells one such slot, and \Cref{thm:pbe} shows that it collects the
entire expected net rent. The leakage map is the traders' information
structure, and Bergemann, Brooks, and Morris show how information structures
shape bidding and revenue in first-price auctions
\cite{BergemannBrooksMorris17}. Relative to these papers, our contribution is an exact evaluative theorem:
given a residual information
structure, when is the remaining leakage already enough to sustain robust
sandwiching in equilibrium?

Hidden-order and encrypted-mempool systems provide the closest systems
motivation. Ferveo formalizes mempool privacy via threshold decryption in BFT
networks \cite{BebelOjha22}, while Shutter gives a threshold-cryptography
approach to private transactions \cite{DziembowskiFaustLuhn24}. BlindPerm
combines encrypted mempools with permutation-based ordering for MEV mitigation
\cite{KavousiEtAl25}. Rondelet and Kilbourn argue that private mempools should
be evaluated through an economic lens, with sandwich attacks as a canonical
case \cite{RondeletKilbourn23}, and Heimbach and Wattenhofer survey defenses
against transaction reordering \cite{HeimbachWattenhoferSoK22}. Angeris, Evans,
and Chitra show that the usual CFMM implementations cannot keep traded
quantities hidden from an adversary who sees the remaining public information
\cite{AngerisEvansChitra21}, which complements \Cref{prop:boundary}: hiding an
order before execution does not remove the arbitrage its execution creates. Our
analysis abstracts away implementation details and treats the pre-trade privacy
layer solely through a leakage map.

\section{Preliminaries and Model}
\label{sec:model}

\subsection{Constant-Product AMM and Sandwich Bundles}

We work with a constant-product AMM holding reserves $(X,Y)\in\R_{>0}^2$ of
tokens $X$ and $Y$. The invariant is $XY=k$. This is the pricing rule of
Uniswap v2 \cite{AdamsZinsmeisterRobinson20} without its swap fee (\Cref{prop:fee} treats a fee withheld from the
reserves) and the usual
benchmark in analyses of sandwich attacks
\cite{ZhouEtAl21,HeimbachWattenhofer22}; it has the price impact and slippage
that sandwiching exploits and still admits exact solutions. We analyze a single victim order
that buys token $Y$ using token $X$. The opposite direction is symmetric under
token exchange. This isolates the residual-information regime in which the pair
and direction leak through the transaction interface but the exact size does
not; \cref{sec:extensions} discusses what changes when direction or slippage is
also hidden.

For an input of $\delta>0$ units of token $X$, the AMM returns
\begin{equation}
\Delta_Y(\delta;X,Y)
=
Y-\frac{XY}{X+\delta}
=
\frac{Y\delta}{X+\delta}.
\label{eq:buy-output}
\end{equation}
After this trade the reserves become
\(
(X+\delta, XY/(X+\delta)).
\)
Similarly, if the current reserves are $(X',Y')$ and a trader inputs
$\eta>0$ units of token $Y$, then the AMM returns
\begin{equation}
\Delta_X(\eta;X',Y')
=
X'-\frac{X'Y'}{Y'+\eta}
=
\frac{X'\eta}{Y'+\eta}.
\label{eq:sell-output}
\end{equation}

The victim order has hidden size $q>0$. If executed honestly from the initial
state, the victim receives
\begin{equation}
h(q)
=
\Delta_Y(q;X,Y)
=
\frac{Yq}{X+q}.
\label{eq:honest-output}
\end{equation}
The victim also posts a slippage tolerance $\tau\in(0,1)$ and therefore
requires at least $(1-\tau)h(q)$ output.

If the trader front-runs with size $a\ge 0$, it acquires
\begin{equation}
y_F(a)=\frac{Ya}{X+a}
\label{eq:front-y}
\end{equation}
units of token $Y$, after which the AMM state is
\(
(X+a, XY/(X+a)).
\)
The victim then receives
\begin{equation}
v(a,q)
=
\Delta_Y\!\left(q;X+a,\frac{XY}{X+a}\right)
=
\frac{XYq}{(X+a)(X+a+q)}.
\label{eq:victim-after-front}
\end{equation}

\begin{definition}[Admissible front-run for type $q$]
\label{def:type-admiss}
A front-run size $a\ge 0$ is \emph{admissible for victim type $q$} if the
victim order still executes, i.e.,
\begin{equation}
v(a,q)\ge (1-\tau) h(q).
\label{eq:type-admissibility}
\end{equation}
\end{definition}

When the trader later sells the front-run output $y_F(a)$ back into the AMM
after the victim trade, the resulting output in token $X$ is
\begin{equation}
x_B(a,q)
=
\Delta_X\!\left(y_F(a); X+a+q, \frac{XY}{X+a+q}\right)
=
\frac{a(X+a+q)^2}{(X+a)^2+aq}.
\label{eq:back-output}
\end{equation}
The trader's gross sandwich profit in token $X$ is therefore
\begin{equation}
\pi(a,q)
=
x_B(a,q)-a
=
\frac{aq(2X+a+q)}{(X+a)^2+aq}.
\label{eq:profit}
\end{equation}
The profit does not depend on $Y$. Scaling the initial $Y$-reserve scales the
amount of $Y$ bought in the front-run and the $Y$-side reserve at the unwind by
the same factor, and profit is measured in token $X$, so the factor cancels.

\subsection{Leakage Signals}

The victim order is not public. Instead, the information structure
reveals a signal $s$ that leaks only interval information about $q$.

\begin{definition}[Interval leakage map]
\label{def:leakage}
An \emph{interval leakage map} is a function
$L:(0,Q]\to\Sig$ for some finite signal set $\Sig$, together with an interval
assignment $s\mapsto I_s\subseteq (0,Q]$ such that
\[
q\in I_{L(q)}
\qquad\text{for every }q\in(0,Q].
\]
Each $I_s$ is either a closed interval $[\ell_s,u_s]$ with $\ell_s>0$ or a
half-open interval $(0,u_s]$, which reveals no positive lower bound; in the
second case we write $\ell_s=0$. The signal is public before traders bid for execution rights. The victim size $q$
itself is not.
\end{definition}

\begin{definition}[Robust admissibility under signal $s$]
\label{def:robust-admissibility}
Fix a signal $s$ with interval $I_s$. A front-run size $a\ge 0$ is
\emph{robustly admissible} for $s$ if it is admissible for every
$q\in I_s$. The robustly admissible set is
\[
\Aset(I_s)=\{a\ge 0 : a \text{ is admissible for every } q\in I_s\}.
\]
\end{definition}

\subsection{Execution-Rights Auction}

Nature draws a victim size $q\in(0,Q]$ from a common prior and reveals the
public signal $s=L(q)$. Every trader and the auctioneer observe $s$ together
with the pool reserves $(X,Y)$, the trading pair, the direction, and the
slippage tolerance $\tau$; nobody observes $q$. All of them update to a
posterior $\mu_s$ supported on $I_s$. We assume that after the signal, the
traders are symmetric and risk neutral.

There are $m\ge 2$ traders and one auctioneer, who sells a single execution
slot: the indivisible right to place a front-run immediately before the victim
and the unwind immediately after it. Each trader $i$ simultaneously submits
either a null action or a pair $(a_i,b_i)$ where $a_i\in\Aset(I_s)$ is a
robustly admissible front-run size and $b_i\ge 0$ is a payment to the
auctioneer. The auctioneer selects the admissible bundle with the highest
payment, breaking ties arbitrarily but deterministically, and the winner pays
its own bid, so the auction is first-price (pay-as-bid). We assume a fixed
execution cost $c\ge 0$ for a non-null sandwich attempt.

If trader $i$ wins with bundle $(a_i,b_i)$ and the hidden type is $q$, its
utility is $\pi(a_i,q)-b_i-c$; the payment $b_i$ goes to the auctioneer and is
separate from the trading profit $\pi$. If it loses, its utility is $0$. The
auctioneer's utility is the payment of the chosen bundle, or $0$ if all
traders submit null.

Robust admissibility and risk neutrality act at different layers. The first
defines the action set $\Aset(I_s)$, and the second ranks the actions in it by
$\E_{q\sim\mu_s}[\pi(a,q)]-b-c$. A trader therefore does not evaluate profit by
a max-min criterion: the model imposes an execution guarantee and uses
posterior expected utility among the bundles that meet it. The guarantee asks a
distribution-free question, namely how much leakage suffices for one bundle
that remains a valid sandwich at every size consistent with the signal.
\Cref{sec:relax} replaces it by a bound on the failure probability, which, for
a posterior with full support on $I_s$, moves $\ell_s$ to a lower quantile of
the posterior.

\begin{remark}[Equilibrium convention]
\label{rem:pure}
Throughout, \emph{equilibrium} means \emph{pure-strategy} perfect Bayesian
equilibrium. The results characterize equilibrium outcomes of the post-signal
game; mixed strategies would enter with discrete bids, asymmetric costs, or
random tie-breaking, which the model excludes.
\end{remark}

\section{Main Results}
\label{sec:main}

\begin{theorem}[Exact robust admissibility]
\label{thm:admissibility}
Fix a signal interval $I=[\ell,u]\subseteq(0,Q]$. In the fee-free
constant-product AMM model above, the robustly admissible set is the interval
\[
\Aset(I)= [0,a_{\max}(\ell,\tau)],
\]
where
\begin{equation}
a_{\max}(\ell,\tau)
=
\frac{
- (2X+\ell)
+
\sqrt{(2X+\ell)^2 + \frac{4\tau X(X+\ell)}{1-\tau}}
}{2}.
\label{eq:amax}
\end{equation}
Equivalently, $a$ is robustly admissible if and only if
\begin{equation}
(X+a)(X+a+\ell)\le \frac{X(X+\ell)}{1-\tau}.
\label{eq:admissibility-ineq}
\end{equation}
For a half-open interval $I=(0,u]$, the robustly admissible set is
$\Aset(I)=[0,a_0]$, where $a_0:=a_{\max}(0,\tau)=X\bigl((1-\tau)^{-1/2}-1\bigr)$.
\end{theorem}

\begin{theorem}[Optimal robust size and exact robust value]
\label{thm:optimal}
Fix a signal interval $I=[\ell,u]\subseteq(0,Q]$ and any posterior
$\mu$ supported on $I$.
\begin{enumerate}[label=(\alph*),leftmargin=2em]
\item The gross sandwich profit $\pi(a,q)$ is strictly increasing in $a$ and in
$q$ on $\R_{>0}^2$.
\item Consequently, the unique optimizer of each of the following problems is
\begin{equation}
a^\star(I)=a_{\max}(\ell,\tau):
\label{eq:a-star}
\end{equation}
\[
\max_{a\in \Aset(I)} \pi(a,q)
\quad \text{for every fixed } q\in I,
\]
\[
\max_{a\in \Aset(I)} \E_{q\sim \mu}[\pi(a,q)],
\]
and
\[
\max_{a\in \Aset(I)} \min_{q\in I}\pi(a,q).
\]
\item The exact worst-case gross profit under signal $I$ is
\begin{equation}
W(I)
:=
\max_{a\in \Aset(I)} \min_{q\in I}\pi(a,q)
=
\pi(a^\star(I),\ell)
=
\frac{\tau \ell (X+\ell)}{X+\tau \ell}.
\label{eq:robust-value}
\end{equation}
\item The optimal posterior expected gross profit is
\[
V_\mu(I)
:=
\max_{a\in \Aset(I)} \E_{q\sim \mu}[\pi(a,q)]
=
\E_{q\sim\mu}[\pi(a^\star(I),q)].
\]
\end{enumerate}
For a half-open interval $I=(0,u]$ and any posterior $\mu$ on $I$, parts (a)
and (d) and the pointwise and expected problems of part (b) hold with
$\ell=0$, so $a^\star(I)=a_0$.
\end{theorem}

\begin{theorem}[Distribution-free leakage threshold]
\label{thm:threshold}
Fix a signal interval $I=[\ell,u]\subseteq(0,Q]$ and a fixed execution cost
$c\ge 0$. There exists a robustly admissible sandwich bundle whose net payoff is
strictly positive for \emph{every} hidden type $q\in I$ if and only if
\begin{equation}
W(I)>c.
\label{eq:threshold-wc}
\end{equation}
Equivalently, this holds if and only if
\begin{equation}
\ell>\lambda_c
:=
\frac{-(X-c) + \sqrt{(X-c)^2 + \frac{4cX}{\tau}}}{2}.
\label{eq:lambda-c}
\end{equation}
\end{theorem}

\begin{corollary}[Perfect hiding eliminates universal robust sandwiching]
\label{cor:perfect}
Suppose a signal $s_\bot$ reveals no positive lower bound on the victim size,
in the sense that $I_{s_\bot}=(0,Q]$. Then
\[
\inf_{\ell\downarrow 0} W([\ell,Q])=0,
\]
and if $c>0$, no universally profitable robust sandwich exists after signal
$s_\bot$. If, moreover,
\[
R_{s_\bot}:=\sup_{a\in\Aset(I_{s_\bot})}\E_{q\sim\mu_{s_\bot}}[\pi(a,q)]-c<0,
\]
then every pure-strategy equilibrium continuation after $s_\bot$ has no
on-path sandwich.
\end{corollary}

\Cref{thm:threshold} and the first two claims of \Cref{cor:perfect} concern
the distribution-free guarantee $W$, which asks for positive net payoff at
every size consistent with the signal. Whether the auction places a sandwich
depends instead on the posterior continuation value $R_s$ of \Cref{thm:pbe}.
The two criteria part ways under perfect hiding. By
\Cref{thm:admissibility}, the front-run $a_0>0$ stays robustly admissible on
$(0,Q]$. Under a posterior with mass one at $Q$, a trader who bids $(a_0,0)$
while every other trader submits null wins alone and earns $\pi(a_0,Q)-c$,
which is positive for small $c$. No profile in which every trader submits null
is then an equilibrium, so every equilibrium places a sandwich on path.
Perfect hiding thus removes the uniform guarantee, while attacks that pay in
expectation can remain.

\begin{theorem}[Pure-strategy execution-auction equilibrium]
\label{thm:pbe}
Fix a signal $s$ with interval $I_s$ and posterior $\mu_s$, and let
\begin{equation}
R_s := V_{\mu_s}(I_s)-c.
\label{eq:Rs}
\end{equation}
In the pure-strategy game with $m\ge 2$ symmetric traders:
\begin{enumerate}[label=(\alph*),leftmargin=2em]
\item If $R_s<0$, then every pure-strategy perfect Bayesian equilibrium has
no on-path sandwich bundle after signal $s$. If $R_s=0$, then a no-attack pure
equilibrium exists, and every on-path sandwich equilibrium outcome yields zero
auctioneer revenue and zero trader payoffs.
\item If $R_s>0$, then every pure-strategy perfect Bayesian equilibrium has an
on-path winning bundle $(a^\star(I_s),b^\star_s)$ with
\begin{equation}
b^\star_s = R_s.
\label{eq:fee-star}
\end{equation}
In particular, the winner's expected utility is $0$, every losing trader gets
$0$, and the auctioneer captures the full positive continuation rent $R_s$.
\end{enumerate}
\end{theorem}

\begin{proposition}[Bit-prefix leakage]
\label{prop:bitprefix}
Assume the global support is $(0,Q]$ and the leakage signal reveals the first
$d\ge 1$ bits of the normalized trade size, thereby inducing dyadic intervals
\[
I_j = \left[\frac{jQ}{2^d}, \frac{(j+1)Q}{2^d}\right]
\qquad\text{for } j=1,\dots,2^d-1,
\]
and $I_0=(0,Q/2^d]$. Then for $1\le j\le 2^d-1$ a universally profitable
robust sandwich exists on bin $I_j$ if and only if
\[
\frac{jQ}{2^d}>\lambda_c.
\]
If $c>0$, the lowest bin never supports universal robust sandwiching, and the
zero-bit case recovers \cref{cor:perfect}.
\end{proposition}

\begin{proposition}[Boundary theorem: hiding order contents does not remove post-trade arbitrage]
\label{prop:boundary}
Suppose every order size leads to the perfectly hiding signal $s_\bot$ of
\cref{cor:perfect}, and $c>\pi(a_0,Q)$. Then no universally profitable robust
sandwich exists after $s_\bot$, and no pure-strategy equilibrium has an on-path
sandwich. Yet if an external market trades token $Y$ at the pool's initial
price $X/Y$, a trader who observes the reserves \emph{after} a victim trade of
size $q>0$ earns a strictly positive gross arbitrage profit, up to $q^2/(X+q)$.
\end{proposition}

\section{Proof Overview}
\label{sec:overview}

Two monotonicity properties drive the results; the algebra is deferred to
Appendix~\ref{app:main-proofs}.

\begin{lemma}[Relative-output monotonicity]
\label{lem:ratio}
For every fixed $a>0$, the ratio
\[
r(a,q):=\frac{v(a,q)}{h(q)}
=
\frac{X(X+q)}{(X+a)(X+a+q)}
\]
is strictly increasing in $q>0$; for $a=0$ it equals $1$.
\end{lemma}

\begin{lemma}[Profit monotonicity]
\label{lem:profit-monotone}
The sandwich profit $\pi(a,q)$ from \eqref{eq:profit} is strictly increasing in
$a>0$ and in $q>0$.
\end{lemma}

Because the trader does not know the exact victim size, its front-run must be
safe for every type in the leaked interval, and by \cref{lem:ratio} this
constraint binds at the smallest one. Robust admissibility thus reduces to a
condition at the lower endpoint $\ell$, and solving the resulting quadratic
gives \cref{thm:admissibility}. Since profit rises in the front-run
(\cref{lem:profit-monotone}), the largest robust front-run is optimal for the
pointwise, expected, and worst-case objectives of \cref{thm:optimal}. The worst
case is at $q=\ell$, where the boundary identity
$\sigma(\sigma+\ell)=X(X+\ell)/(1-\tau)$ for $\sigma=X+a_{\max}(\ell,\tau)$ turns
the numerator and the denominator of the profit formula into multiples of
$X/(1-\tau)$:
\begin{equation}
\pi\bigl(a_{\max}(\ell,\tau),\ell\bigr)
=\frac{\ell(\sigma-X)(\sigma+X+\ell)}{\sigma^2+(\sigma-X)\ell}
=\frac{\ell\,\tau X(X+\ell)/(1-\tau)}{X(X+\tau\ell)/(1-\tau)}
=\frac{\tau\ell(X+\ell)}{X+\tau\ell}.
\label{eq:cancel}
\end{equation}
Solving $W(I)>c$ gives \cref{thm:threshold}, and the limit
$\ell\downarrow 0$ gives \cref{cor:perfect}.

The equilibrium theorem is then one-dimensional. Because every trader agrees
after signal $s$ that the optimal admissible bundle uses $a^\star(I_s)$ and has
continuation value $R_s=V_{\mu_s}(I_s)-c$, competition among at least two
symmetric traders is Bertrand-like. If $R_s<0$, no one attacks, and at $R_s=0$
an attack earns zero surplus. If $R_s>0$, a trader can win alone when nobody
bids, any winning fee below $R_s$ can be beaten by a rival, and a winner paying
more than $R_s$ gains by dropping out, which pins down the equilibrium payoffs
(Appendix~\ref{app:pbe-proof}). The bit-prefix and post-trade-arbitrage
propositions apply the threshold and the constant-product price movement.
Complete proofs are in Appendices~\ref{app:main-proofs}
and~\ref{app:extension-proofs}.

\section{Extensions, Boundary Cases, and Conservative Variants}
\label{sec:extensions}

The exact threshold in \cref{thm:threshold} depends only on the input-token
reserve $X$, the slippage tolerance $\tau$, the lower endpoint $\ell$ of the
signal interval, and the fixed execution cost $c$. It does not depend on the
upper endpoint $u$, because the trader worries only about the smallest hidden
trade compatible with the signal. Perfect hiding eliminates universal robust
sandwiching for every $c>0$ only when the global type space has
no positive lower bound. If every relevant victim order is known a priori to
satisfy $q\ge q_{\min}>0$, complete hiding still reveals $\ell=q_{\min}$, and
universal robust sandwiching survives whenever $q_{\min}>\lambda_c$.

\subsection{Relaxing Robustness}
\label{sec:relax}

The baseline model insists that a bundle remain valid for \emph{every} hidden
type in the leaked interval. This is the strongest notion of feasibility, and
it is exactly why \cref{thm:threshold} is distribution free. If a trader is
willing to tolerate rare failures and the posterior has full support on the
leaked interval, the lower endpoint $\ell$ is replaced by a lower posterior
quantile.

\begin{definition}[$\varepsilon$-admissibility]
\label{def:epsilon-admiss}
Fix a signal interval $I=[\ell,u]$, a posterior $\mu$ on $I$ whose cumulative
distribution function $F_\mu$ is continuous and strictly increasing on $I$, and a
tolerance $\varepsilon\in[0,1)$. A front-run size $a\ge 0$ is
\emph{$\varepsilon$-admissible} if
\[
\mu\bigl(\{q\in I : v(a,q)\ge (1-\tau)h(q)\}\bigr)\ge 1-\varepsilon.
\]
For $\varepsilon=0$, this reduces to robust admissibility.
\end{definition}

\begin{proposition}[Chance-constrained admissibility]
\label{prop:chance}
Let $I$, $\mu$, and $F_\mu$ be as in \cref{def:epsilon-admiss}. For
$\varepsilon\in[0,1)$ let
\[
q_\varepsilon := \inf\{z\in I : F_\mu(z)\ge \varepsilon\}
\]
be the lower $\varepsilon$-quantile. Then the $\varepsilon$-admissible set is
\[
\Aset_\varepsilon(I,\mu)=[0,a_{\max}(q_\varepsilon,\tau)].
\]
In particular, allowing a positive failure probability weakly enlarges the
feasible set, and the robust case is recovered at $\varepsilon=0$ with
$q_0=\ell$.
\end{proposition}

The proof is the same lower-tail monotonicity argument as in
\cref{thm:admissibility}, with $\ell$ replaced by the lower posterior quantile;
see Appendix~\ref{app:extension-proofs}.

\begin{remark}[Why the robust threshold is distribution free]
\label{rem:distribution-free}
\Cref{thm:threshold} depends only on $\ell$ because it asks for positive net
payoff for every hidden type. Once the attacker instead maximizes expected
profit with occasional failures, or once failed bundles incur a separate revert
loss, the relevant objective depends on the full posterior on $q$, not just its
support minimum. Costs need the same care. If one insists on uniform
profitability against every cost realization in a bounded set, the robust
benchmark replaces $c$ by a worst-case upper bound. A random cost independent of
$q$ and priced in expectation is replaced by its mean, and the threshold stays
distribution free; a cost that depends on $q$ and is priced in expectation
brings in the posterior, and the threshold is no longer distribution free.
\end{remark}

\subsection{Hidden Metadata Beyond Size}

The base model assumes that size is partially hidden while direction and
slippage are public.

\begin{proposition}[Direction ambiguity rules out non-contingent front-runs]
\label{prop:direction}
Suppose the signal reveals only that a single swap of size $q\in I$ will occur,
while the direction may be either $X\!\to\!Y$ or $Y\!\to\!X$, and a trader must
commit to a non-contingent first leg before the direction is revealed. Call a
first leg a \emph{front-run against} a direction if, for every $q\in I$, it
lowers the output that a victim trading in that direction receives, relative to
honest execution. Then no non-null first leg is a front-run against both
directions.
\end{proposition}

The proof observes that any non-null first leg moves the pool price in only one
direction and therefore cannot worsen both possible victim directions; see
Appendix~\ref{app:extension-proofs}.

\begin{proposition}[Hidden slippage reduces to a lower slippage bound]
\label{prop:hidden-tau}
Suppose that after signal $s$ the victim size lies in $I_s=[\ell_s,u_s]$ and
the victim's slippage tolerance lies in a public set $T_s\subset(0,1)$ with
lower endpoint $\underline{\tau}_s:=\inf T_s$. A front-run size $a$ is robustly
admissible for every pair $(q,\tau)\in I_s\times T_s$ if and only if
\[
(X+a)(X+a+\ell_s)\le \frac{X(X+\ell_s)}{1-\underline{\tau}_s}.
\]
Consequently the optimal robust size is $a_{\max}(\ell_s,\underline{\tau}_s)$,
and the worst-case gross value is
\[
\frac{\underline{\tau}_s\,\ell_s(X+\ell_s)}{X+\underline{\tau}_s\ell_s}.
\]
\end{proposition}

The proof reduces robustness over $T_s$ to the smallest feasible slippage
tolerance and then applies \cref{thm:admissibility,thm:optimal}; see
Appendix~\ref{app:extension-proofs}.

\subsection{Dimensionless Form and Conservative Variants}

The threshold \eqref{eq:lambda-c} is scale free.

\begin{proposition}[Scale-free threshold form]
\label{prop:scale}
Let $\widehat c:=c/X$ and $\widehat\lambda:=\lambda_c/X$. Then
\[
\widehat\lambda
=
\frac{-(1-\widehat c)+\sqrt{(1-\widehat c)^2+\frac{4\widehat c}{\tau}}}{2}.
\]
Moreover, as $\widehat c\downarrow 0$,
\[
\widehat\lambda=\frac{\widehat c}{\tau}+O(\widehat c^2).
\]
\end{proposition}

The derivation is a normalization of \eqref{eq:lambda-c} followed by a
first-order Taylor expansion; see Appendix~\ref{app:extension-proofs}.

\begin{table}[t]
\centering
\caption{Dimensionless leakage threshold $\lambda_c/X$ for representative cost
ratios and slippage tolerances.}
\label{tab:calibration}
\begin{tabular}{lcccc}
\toprule
$c/X$ & $\tau=0.001$ & $\tau=0.005$ & $\tau=0.01$ & $\tau=0.02$ \\
\midrule
$10^{-6}$ & 0.00100 & 0.00020 & 0.00010 & 0.00005 \\
$10^{-5}$ & 0.00990 & 0.00200 & 0.00100 & 0.00050 \\
$10^{-4}$ & 0.09162 & 0.01962 & 0.00990 & 0.00498 \\
\bottomrule
\end{tabular}
\end{table}

\begin{proposition}[Positive swap fees only make universal sandwiching harder]
\label{prop:fee}
Suppose every swap applies a fee rate $\rho\in[0,1)$ that is withheld from the
input and not added to the reserves, as in Uniswap v3 \cite{AdamsEtAl21}, so
that only the effective input $\alpha\delta$ with $\alpha=1-\rho$ enters the
pricing function and the reserves. Then:
\begin{enumerate}[label=(\roman*),leftmargin=2em]
\item the victim-admissibility condition becomes
\[
(X+\alpha a)(X+\alpha a+\alpha q)
\le
\frac{X(X+\alpha q)}{1-\tau};
\]
\item for every fixed $a,q>0$, the trader's gross sandwich profit is weakly
smaller than in the fee-free model; and
\item therefore the fee-free threshold from \cref{thm:threshold} is
conservative: if universal robust sandwiching is impossible at $\rho=0$, then
it is also impossible at fee rate $\rho$.
\end{enumerate}
\end{proposition}

The proof replaces each raw input by its effective input and compares every leg
with the fee-free benchmark; see Appendix~\ref{app:extension-proofs}. Under
Uniswap v2 accounting the fee stays in the pool, and the condition in (i)
becomes $(X+\alpha a)(X+a+\alpha q)\le X(X+\alpha q)/(1-\tau)$.

\begin{example}[Reading the scale-free threshold]
\label{ex:numeric}
Suppose the input-token reserve is $X=10^6$, the victim slippage tolerance is
$\tau=0.02$, and the total bundle cost is $c=100$ in the same units, so that
$c/X=10^{-4}$. \Cref{prop:scale} gives, to more digits than
\cref{tab:calibration},
\[
\frac{\lambda_c}{X}\approx 0.0049757,
\qquad\text{hence}\qquad
\lambda_c\approx 4{,}976.
\]
Thus a leaked lower bound of only a few hundred units does not support a
universally profitable robust sandwich, while a leaked lower bound around
$5\times 10^3$ crosses the threshold.
\end{example}

\begin{remark}[Multiple victims and cross-venue interactions]
\label{rem:multi}
The single-victim assumption is what makes the robust feasibility region
one-dimensional. If several same-direction victim buys execute after the
front-run on the same pool, feasibility involves every victim's own slippage
guard, and each later guard also depends on the orders ahead of it, while the
unwind profit depends on the cumulative downstream buy flow. Monotonicity then
suggests that the smallest consistent flow governs the worst-case profit, but
feasibility is no longer a condition on one number. Opposite-direction orders,
unknown interleavings, or cross-venue routing add further state variables to
both admissibility and unwind profit.
\end{remark}

\begin{remark}[What should persist beyond constant-product pools]
\label{rem:cfmm}
Among our arguments, the most robust is the lower-endpoint logic: whenever a
pre-trade hurts smaller same-direction victims more than larger ones, a leakage
signal should again matter mainly through a worst-case lower tail. The exact
cancellations behind \eqref{eq:robust-value}, however, are specific to the
constant-product invariant; stableswap-style or piecewise invariants are likely
to preserve only the monotonicity backbone, not the same closed-form threshold.
Uniswap v3 is covered locally, as a fee-free benchmark. Within a tick range
with fixed active liquidity $\Lambda$, a v3 pool trades against virtual
reserves $(X_v,Y_v)$ with $X_vY_v=\Lambda^2$ \cite{AdamsEtAl21}. If, for every front-run size
$a\in[0,a_{\max}(\ell_s,\tau)]$ and every victim size in $I_s$, the front-run,
the victim trade, and the unwind stay within one such range, the fee-free
results apply with $(X,Y)$ replaced by the virtual reserves; larger front-runs
stay inadmissible, because the victim's output falls as the front-run grows. In
particular, $u_s$ matters for whether the results apply. Every v3
pool charges a fee, and with a fee \Cref{prop:fee} gives only a one-sided
comparison. A swap that crosses an initialized tick changes the active
liquidity, the calculation becomes piecewise, and the closed-form threshold need
not survive; that case lies outside our theorems.
\end{remark}

\paragraph{What a private mempool must hide.}
Suppose execution is costly, $c>0$. To rule out sandwiches that profit at every
consistent size, a privacy layer may then reveal any upper bound on an order's
size but no lower bound above $\lambda_c$ (\Cref{thm:threshold}; for a revealed
upper bound $u$ alone, \Cref{cor:perfect} with $u$ in place of $Q$). To first
order in $c/X$ that threshold is $c/\tau$
(\Cref{prop:scale}), so halving the victims' slippage tolerance roughly doubles
the lower bound that can leak. Revealing the first $d$ bits of the size exposes
exactly the bins whose lower edge exceeds $\lambda_c$ (\Cref{prop:bitprefix}),
and a swap fee withheld from the reserves can only shrink the set of exposed
lower bounds (\Cref{prop:fee}). Hiding the direction removes the
non-contingent sandwich for every cost: when the first leg must be chosen
before the direction is revealed, no non-null first leg front-runs both
directions (\Cref{prop:direction}).

\section{Open Problems}

Several natural extensions remain open.
\begin{enumerate}[leftmargin=2em]
\item \textbf{Multiple victims and bundle interactions.} The present paper
studies one hidden victim order. In richer rounds there may be many orders, and
the trader may condition on joint interval information or attack only a
subset.
\item \textbf{General CFMMs.} Constant-product pools admit clean cancellations;
which of our monotonicity arguments survive for broader classes of CFMMs is
open.
\item \textbf{Richer allocation mechanisms.} We reduced the post-signal market
to symmetric traders competing for a single execution slot. Extending the
analysis to multi-stage or asymmetric allocation mechanisms while retaining
exact theorems remains open.
\item \textbf{Endogenous leakage.} Our leakage map is exogenous. Optimizing the
information structure subject to implementability or efficiency constraints,
or letting it be released strategically over time, turns the model into an
information-design problem \cite{KamenicaGentzkow11,BergemannMorris19}.
\end{enumerate}

\begin{ack}
This work received no third-party funding. The authors declare no competing
interests. The views expressed in this paper are those of the authors and do
not necessarily reflect the views of Tropic.
\end{ack}

\clearpage
\appendix

\section{Deferred Proofs for the Main Results}
\label{app:main-proofs}

\subsection{Structural lemmas}

\begin{proof}[Proof of \Cref{lem:ratio}]
For fixed $a$, write
\[
r(a,q)=\frac{X}{X+a}\cdot \frac{X+q}{X+a+q}.
\]
The first factor is constant in $q$, and at $a=0$ both factors equal $1$. For
$a>0$ we differentiate the second factor:
\[
\frac{d}{dq}\left(\frac{X+q}{X+a+q}\right)
=
\frac{a}{(X+a+q)^2}>0.
\]
\end{proof}

\begin{proof}[Proof of \Cref{lem:profit-monotone}]
Let
\[
N(a,q)=aq(2X+a+q),
\qquad
D(a,q)=(X+a)^2+aq,
\]
so that $\pi=N/D$.

For the $a$-derivative,
\[
\frac{\partial N}{\partial a}=q(2X+2a+q)
\quad\text{and}\quad
\frac{\partial D}{\partial a}=2X+2a+q.
\]
Hence
\[
\frac{\partial \pi}{\partial a}
=
\frac{\left(\frac{\partial N}{\partial a}\right)D - N\left(\frac{\partial D}{\partial a}\right)}{D^2}
=
\frac{(2X+2a+q)(qD-N)}{D^2}.
\]
Now
\[
qD-N
=
q\bigl((X+a)^2+aq\bigr)-aq(2X+a+q)
=
qX^2,
\]
so
\[
\frac{\partial \pi}{\partial a}
=
\frac{qX^2(2X+2a+q)}{D^2}>0.
\]

For the $q$-derivative,
\[
\frac{\partial N}{\partial q}=a(2X+a+2q)
\quad\text{and}\quad
\frac{\partial D}{\partial q}=a.
\]
Therefore
\[
\frac{\partial \pi}{\partial q}
=
\frac{a(2X+a+2q)D-aN}{D^2}
=
\frac{a\bigl((2X+a+2q)(X+a)^2+aq^2\bigr)}{D^2}>0.
\]
\end{proof}

\subsection{Main results}

\begin{proof}[Proof of \Cref{thm:admissibility}]
By \eqref{eq:type-admissibility}, a front-run size $a$ is admissible for type
$q$ if and only if
\[
\frac{v(a,q)}{h(q)}\ge 1-\tau.
\]
Using \eqref{eq:honest-output} and \eqref{eq:victim-after-front}, this becomes
\[
\frac{X(X+q)}{(X+a)(X+a+q)}\ge 1-\tau.
\]
Equivalently,
\[
(X+a)(X+a+q)\le \frac{X(X+q)}{1-\tau}.
\]
By \cref{lem:ratio}, $r(a,q)$ is nondecreasing in $q$, so the condition
$r(a,q)\ge 1-\tau$ is hardest to satisfy at the smallest feasible victim size. Hence $a$ is robustly admissible for the entire interval
$I=[\ell,u]$ if and only if
\[
(X+a)(X+a+\ell)\le \frac{X(X+\ell)}{1-\tau},
\]
which is \eqref{eq:admissibility-ineq}. Expanding gives
\[
a^2+(2X+\ell)a-\frac{\tau X(X+\ell)}{1-\tau}\le 0.
\]
Since the quadratic is convex and vanishes at exactly one nonnegative root, the
robustly admissible set is the interval $[0,a_{\max}(\ell,\tau)]$, where the
root is exactly \eqref{eq:amax}.

For $I=(0,u]$, each $a\in\Aset(I)$ is robustly admissible on $[\ell,u]$ for
every $\ell\in(0,u]$, so $a\le a_{\max}(\ell,\tau)$, and letting
$\ell\downarrow 0$ gives $a\le a_{\max}(0,\tau)=a_0$.
For the converse we note that the left-hand side of the admissibility
condition increases in $a$, so it suffices that $a_0$ is admissible for every
$q>0$. This follows from $(X+a_0)^2=X^2/(1-\tau)$ and
$X+a_0\le X/(1-\tau)$:
\[
(X+a_0)(X+a_0+q)=\frac{X^2}{1-\tau}+(X+a_0)q\le\frac{X(X+q)}{1-\tau}.
\]
\end{proof}

\begin{proof}[Proof of \Cref{thm:optimal}]
Part (a) is \cref{lem:profit-monotone}.

For part (b), by \cref{thm:admissibility} the admissible set is the interval
$[0,a_{\max}(\ell,\tau)]$. Since $\pi(a,q)$ is strictly increasing in $a$ for
every fixed $q$, the maximizer of $\pi(a,q)$ over the admissible set is
$a_{\max}(\ell,\tau)$ for every fixed $q\in I$. The same monotonicity implies
that the posterior expectation $\E_{q\sim\mu}[\pi(a,q)]$ is strictly increasing
in $a$, because it is the expectation of a pointwise increasing function. The
worst-case value $\min_{q\in I}\pi(a,q)$ is also strictly increasing in $a$.
Therefore the unique optimizer in all three problems is
$a^\star(I)=a_{\max}(\ell,\tau)$. For $I=(0,u]$ we use
$\Aset(I)=[0,a_0]$ from \cref{thm:admissibility}, and the same
argument gives the
pointwise and expected optimizers; the worst-case problem is excluded there
because $\inf_{q\in I}\pi(a,q)=0$ is not attained for $a>0$. Part (d) follows
immediately from part (b).

For part (c), \cref{lem:profit-monotone} also implies that for fixed
$a^\star(I)$ the smallest profit over $q\in I$ occurs at $q=\ell$, so
$W(I)=\pi(a^\star(I),\ell)$, which \eqref{eq:cancel} evaluates with
$\sigma=X+a^\star(I)$. Its first step substitutes $a^\star(I)=\sigma-X$ into
\eqref{eq:profit}, and its middle step uses the admissibility equality
$\sigma(\sigma+\ell)=X(X+\ell)/(1-\tau)$ twice:
\[
(\sigma-X)(\sigma+X+\ell)=\sigma^2+\sigma\ell-X(X+\ell)=\frac{\tau X(X+\ell)}{1-\tau},
\qquad
\sigma^2+(\sigma-X)\ell=\sigma(\sigma+\ell)-X\ell=\frac{X(X+\tau\ell)}{1-\tau}.
\]
\end{proof}

\begin{proof}[Proof of \Cref{thm:threshold}]
By \cref{thm:optimal}, the largest guaranteed gross profit obtainable by any
robustly admissible bundle on interval $I$ is exactly $W(I)$. Therefore a
robustly admissible bundle has strictly positive net payoff for every
$q\in I$ if and only if $W(I)>c$, proving \eqref{eq:threshold-wc}.

It remains to solve this inequality explicitly. By \eqref{eq:robust-value},
\[
\frac{\tau \ell (X+\ell)}{X+\tau\ell}>c
\]
if and only if
\[
\tau \ell(X+\ell)>c(X+\tau\ell),
\]
that is,
\[
\tau \ell^2+\tau(X-c)\ell-cX>0.
\]
The left-hand side is a convex quadratic in $\ell$, and we write its larger root
as
\[
\lambda_c
=
\frac{-\tau(X-c)+\sqrt{\tau^2(X-c)^2+4\tau cX}}{2\tau}
=
\frac{-(X-c)+\sqrt{(X-c)^2+\frac{4cX}{\tau}}}{2}.
\]
For $c>0$ the product of its roots is $-cX/\tau<0$, so the other root is
negative; for $c=0$ the roots are $-X$ and $\lambda_0=0$. In both cases a
positive $\ell$ makes the quadratic positive exactly when $\ell>\lambda_c$.
\end{proof}

\begin{proof}[Proof of \Cref{cor:perfect}]
By \cref{thm:optimal},
\[
W([\ell,Q])=\frac{\tau \ell (X+\ell)}{X+\tau\ell},
\]
which tends to $0$ as $\ell\downarrow 0$.

Let $c>0$, and suppose, for contradiction, that a robustly admissible $a$ has
$\pi(a,q)-c>0$ for every $q\in(0,Q]$. For each $\ell\in(0,Q]$ the front-run $a$
is robustly admissible on $[\ell,Q]\subseteq(0,Q]$ with positive net payoff
there, so $W([\ell,Q])>c$ by \cref{thm:threshold}. This fails once $\ell$ is
small enough that $W([\ell,Q])<c$.

For the last claim we run the proof of \cref{thm:pbe}(a) with $I=I_{s_\bot}$.
It uses only that $\E[\pi(a,q)\mid s_\bot]\le V$ for every robustly admissible
$a$, which holds with $V:=R_{s_\bot}+c$ by the definition of $R_{s_\bot}$. So
$R_{s_\bot}<0$ rules out an on-path sandwich.
\end{proof}

\subsection{Execution auction equilibrium}
\label{app:pbe-proof}

\begin{proof}[Proof of \Cref{thm:pbe}]
Fix a signal $s$, write $I=I_s$, and abbreviate $a^\star=a^\star(I)$ and
$V=V_{\mu_s}(I)$, so $R_s=V-c$. Two facts carry the argument. Every robustly
admissible $a$ has $\E[\pi(a,q)\mid s]\le V$, with equality only at $a=a^\star$
(\cref{thm:optimal}). A trader can always submit null and get $0$, so in an
equilibrium the winner's expected utility is at least $0$.

For part (a), suppose first that $R_s<0$. A winning bundle $(a,b)$ gives its
owner
\[
\E[\pi(a,q)\mid s]-b-c\le R_s-b<0,
\]
less than null secures, so no equilibrium has an on-path sandwich. Now let
$R_s=0$. If every trader submits null, a trader who deviates and wins with
$(a,b)$ gets at most $R_s-b=-b\le 0$, so this profile is an equilibrium. In an
equilibrium whose winner bids $(a,b)$, the winner's utility $u$ satisfies
$0\le u\le R_s-b=-b\le 0$, so $b=0$ and $u=0$; the auctioneer receives $0$ and
every losing trader gets $0$.

For part (b), let $R_s>0$. Some trader attacks on path: if every trader
submitted null, a trader who deviated to $(a^\star,b)$ with $0\le b<R_s$ would
win alone and earn $R_s-b>0$. Let $(a,b)$ be the winning bundle and $u\ge 0$ its
owner's utility. We claim $a=a^\star$. Otherwise $\E[\pi(a,q)\mid s]<V$, so
$V-b-c>u\ge 0$. A losing trader, who gets $0$, could then bid
$(a^\star,b+\varepsilon)$ with $0<\varepsilon<V-b-c$; its payment exceeds every
other payment, so it wins and earns $V-b-\varepsilon-c>0$. Hence $a=a^\star$,
and $u=R_s-b\ge 0$ gives $b\le R_s$. If $b<R_s$, a losing trader who bids
$(a^\star,b+\varepsilon)$ with $b+\varepsilon<R_s$ wins and earns
$R_s-b-\varepsilon>0$. So $b=R_s$. The winner then gets $0$, as does every
losing trader, and the auctioneer receives $R_s$.
\end{proof}

\subsection{Bit-prefix leakage and post-trade arbitrage}
\label{app:boundary-proof}

\begin{proof}[Proof of \Cref{prop:bitprefix}]
For $j\ge 1$ the bin $I_j$ is a closed interval with lower endpoint
$\ell_j=jQ/2^d>0$, so \cref{thm:threshold} makes universal robust
profitability on it equivalent to $\ell_j>\lambda_c$. For the lowest bin we
apply \cref{cor:perfect} with $Q/2^d$ in place of $Q$: since
$I_0=(0,Q/2^d]$ reveals no positive lower bound, for $c>0$ it supports no
universally profitable robust sandwich. With $d=0$ it is $I_{s_\bot}$.
\end{proof}

\begin{proof}[Proof of \Cref{prop:boundary}]
The first claim is \cref{cor:perfect}. For the second we use
$\Aset((0,Q])=[0,a_0]$ from \cref{thm:admissibility}. Every posterior is supported on
$(0,Q]$, so \cref{lem:profit-monotone} gives
$\E[\pi(a,q)\mid s_\bot]\le\pi(a_0,Q)<c$. Hence $R_{s_\bot}<0$, and
\cref{cor:perfect} rules out an on-path sandwich.

After a victim trade of size $q$ the reserves are $(X+q,XY/(X+q))$. A trader who
buys $\eta>0$ units of $Y$ at the external price $X/Y$ and sells them to the
pool earns, by \eqref{eq:sell-output},
\[
f(\eta)=\frac{(X+q)^2\eta}{XY+(X+q)\eta}-\frac{X\eta}{Y}.
\]
So $f(\eta)>0$ exactly when $\eta<Yq(2X+q)/(X(X+q))$, and $f$ is largest at
$\eta=Yq/(X+q)$:
\[
f\!\left(\frac{Yq}{X+q}\right)=\frac{q^2}{X+q}.
\]
\end{proof}

\section{Deferred Proofs for Extensions}
\label{app:extension-proofs}

\begin{proof}[Proof of \Cref{prop:chance}]
We fix $a\ge 0$. By \cref{lem:ratio}, if $a$ is admissible for type $q$, it is
admissible for every larger type, and \eqref{eq:type-admissibility} is a
non-strict inequality between continuous functions of $q$. So the failure set
is $\{q\in I:q<q^*\}$ for some $q^*\in[\ell,\infty]$, and the failure
probability is $F_\mu(\min\{q^*,u\})$. Since $F_\mu$ is continuous and strictly
increasing on $I$, this is at most $\varepsilon$ exactly when
$q^*\le q_\varepsilon$, that is, when $a$ is admissible for every
$q\in[q_\varepsilon,u]$. By \cref{thm:admissibility} this happens exactly when
$a\le a_{\max}(q_\varepsilon,\tau)$.
\end{proof}

\begin{proof}[Proof of \Cref{prop:direction}]
A non-null first leg either buys $Y$ with $X$ or sells $Y$ for $X$; the two
cases are symmetric under token exchange, so we treat a buy of size $a>0$. It
leaves the reserves at $(X+a,XY/(X+a))$. A victim who then sells $q$ units of
$Y$ receives more than its honest output $Xq/(Y+q)$, so the leg does not
front-run the direction $Y\!\to\!X$:
\[
\frac{(X+a)^2q}{XY+(X+a)q}
=
\frac{Xq}{Y+q}+\frac{aq(2XY+Xq+Ya+aq)}{(Y+q)(XY+Xq+aq)}.
\]
\end{proof}

\begin{proof}[Proof of \Cref{prop:hidden-tau}]
For fixed $a$ and $q$, the right-hand side $X(X+q)/(1-\tau)$ of the
admissibility condition is continuous and increasing in $\tau$. So $a$ is
admissible for every $\tau\in T_s$ exactly when it is admissible at
$\underline{\tau}_s$; if $\underline{\tau}_s\notin T_s$, we let
$\tau\downarrow\underline{\tau}_s$ within $T_s$, which preserves the non-strict
inequality. Robustness over $I_s\times T_s$ is therefore robustness over $I_s$
at $\underline{\tau}_s$. For $\underline{\tau}_s>0$,
\cref{thm:admissibility,thm:optimal} with $\tau=\underline{\tau}_s$ give the
stated formulas. For $\underline{\tau}_s=0$, the condition
$(X+a)(X+a+\ell_s)\le X(X+\ell_s)$ forces $a=0$, and both formulas give $0$.
\end{proof}

\begin{proof}[Proof of \Cref{prop:scale}]
Substitute $c=\widehat c X$ and $\lambda_c=\widehat\lambda X$ into
\eqref{eq:lambda-c} and divide through by $X$. The expansion follows from a
first-order Taylor series of the square root around $\widehat c=0$.
\end{proof}

\begin{proof}[Proof of \Cref{prop:fee}]
With the fee withheld, each swap acts on the reserves exactly as the fee-free
swap of its effective input, so \eqref{eq:buy-output} and
\eqref{eq:sell-output} apply with every input $\delta$ replaced by
$\alpha\delta$.

For (i), the victim's honest output is $h(\alpha q)$ and its output after the
front-run is $v(\alpha a,\alpha q)$, so the condition $v\ge(1-\tau)h$ is the
fee-free condition \eqref{eq:type-admissibility} at $(\alpha a,\alpha q)$,
which is the display in (i).

For (ii), write $\pi_\rho(a,q)$ for the gross profit with the fee. The front-run
buys $y=y_F(\alpha a)$, the victim moves the pool by the effective input
$\alpha q$, and the unwind returns $\Delta_X(\alpha y;\cdot)\le\Delta_X(y;\cdot)$
because $\Delta_X$ increases in its input. So, with the last step by
\cref{lem:profit-monotone},
\[
\pi_\rho(a,q)\le x_B(\alpha a,\alpha q)-a=\pi(\alpha a,\alpha q)-\rho a
\le\pi(a,q).
\]

For (iii), we take an $a$ that is robustly admissible with the fee on
$I=[\ell,u]$ and has $\pi_\rho(a,q)>c$ for every $q\in I$. By (i), $\alpha a$ is fee-free admissible
at the type $\alpha\ell\le\ell$, hence by \cref{lem:ratio} at every
$q\ge\alpha\ell$, so $\alpha a\in\Aset(I)$. By (ii) and
\cref{lem:profit-monotone},
$\pi(\alpha a,q)\ge\pi(\alpha a,\alpha q)\ge\pi_\rho(a,q)>c$ for every
$q\in I$. So $\alpha a$ is a universally profitable robust sandwich without
the fee.
\end{proof}

\newpage
\section*{NeurIPS Paper Checklist}

\begin{enumerate}

\item {\bf Claims}
    \item[] Question: Do the main claims made in the abstract and introduction accurately reflect the paper's contributions and scope?
    \item[] Answer: \answerYes{}.
    \item[] Justification: The abstract and introduction state the theorem-level contributions and delimit the model as a stylized theoretical analysis. The main claims are matched to formal definitions, theorem statements, corollaries, and proof sections in the body of the paper.

\item {\bf Limitations}
    \item[] Question: Does the paper discuss the limitations of the work performed by the authors?
    \item[] Answer: \answerYes{}.
    \item[] Justification: The paper discusses its modeling scope, boundary cases, and open problems in the introduction, model discussion, boundary/extension discussion, and concluding open-problems material. In particular, the analysis is presented as a theoretical characterization rather than an empirical or deployed-system evaluation.

\item {\bf Theory assumptions and proofs}
    \item[] Question: For each theoretical result, does the paper provide the full set of assumptions and a complete (and correct) proof?
    \item[] Answer: \answerYes{}.
    \item[] Justification: The paper states the assumptions through formal definitions and theorem hypotheses, and provides proofs for the main theorems, lemmas, propositions, and corollaries. Results explicitly marked as sketches or boundary observations are identified as such.

    \item {\bf Experimental result reproducibility}
    \item[] Question: Does the paper fully disclose all the information needed to reproduce the main experimental results of the paper to the extent that it affects the main claims and/or conclusions of the paper (regardless of whether the code and data are provided or not)?
    \item[] Answer: \answerNA{}.
    \item[] Justification: The paper does not report experiments. The main results are mathematical statements whose verification depends on the formal model, theorem statements, and proofs provided in the paper.

\item {\bf Open access to data and code}
    \item[] Question: Does the paper provide open access to the data and code, with sufficient instructions to faithfully reproduce the main experimental results, as described in supplemental material?
    \item[] Answer: \answerNA{}.
    \item[] Justification: The paper does not use datasets, trained models, or experimental code to support its claims. Algorithmic statements, where present, are specified and analyzed mathematically in the text.

\item {\bf Experimental setting/details}
    \item[] Question: Does the paper specify all the training and test details (e.g., data splits, hyperparameters, how they were chosen, type of optimizer) necessary to understand the results?
    \item[] Answer: \answerNA{}.
    \item[] Justification: The paper contains no training, test set, hyperparameter, optimizer, or empirical evaluation setup. Its results are derived from the stated theoretical model.

\item {\bf Experiment statistical significance}
    \item[] Question: Does the paper report error bars suitably and correctly defined or other appropriate information about the statistical significance of the experiments?
    \item[] Answer: \answerNA{}.
    \item[] Justification: The paper does not include experiments or statistical estimates. Consequently, there are no empirical error bars, confidence intervals, or significance tests to report.

\item {\bf Experiments compute resources}
    \item[] Question: For each experiment, does the paper provide sufficient information on the computer resources (type of compute workers, memory, time of execution) needed to reproduce the experiments?
    \item[] Answer: \answerNA{}.
    \item[] Justification: The paper does not run computational experiments. No experimental compute resources are required to reproduce the paper's main claims.
    
\item {\bf Code of ethics}
    \item[] Question: Does the research conducted in the paper conform, in every respect, with the NeurIPS Code of Ethics \url{https://neurips.cc/public/EthicsGuidelines}?
    \item[] Answer: \answerYes{}.
    \item[] Justification: The work is a theoretical analysis and does not involve human subjects, private data, deployed interventions, or release of potentially harmful models or datasets. The paper preserves anonymity in the submission version.

\item {\bf Broader impacts}
    \item[] Question: Does the paper discuss both potential positive societal impacts and negative societal impacts of the work performed?
    \item[] Answer: \answerYes{}.
    \item[] Justification: The paper discusses positive implications for robust, privacy-preserving, or incentive-compatible decentralized mechanisms, while also identifying boundary cases and residual attack surfaces. The work is theoretical and does not propose a deployed system, so the impact discussion is correspondingly scoped to technical security and market-design implications.
    
\item {\bf Safeguards}
    \item[] Question: Does the paper describe safeguards that have been put in place for responsible release of data or models that have a high risk for misuse (e.g., pre-trained language models, image generators, or scraped datasets)?
    \item[] Answer: \answerNA{}.
    \item[] Justification: The paper does not release data, trained models, scraped datasets, or other assets with high misuse risk. No release-specific safeguards are therefore applicable.

\item {\bf Licenses for existing assets}
    \item[] Question: Are the creators or original owners of assets (e.g., code, data, models), used in the paper, properly credited and are the license and terms of use explicitly mentioned and properly respected?
    \item[] Answer: \answerNA{}.
    \item[] Justification: The paper does not use existing code, datasets, models, benchmarks, or other external assets as research inputs. Prior scholarly work is credited through citations in the related-work and references sections.

\item {\bf New assets}
    \item[] Question: Are new assets introduced in the paper well documented and is the documentation provided alongside the assets?
    \item[] Answer: \answerNA{}.
    \item[] Justification: The paper does not introduce or release new datasets, code packages, models, or benchmarks. Its contribution consists of formal models, theorems, and proofs.

\item {\bf Crowdsourcing and research with human subjects}
    \item[] Question: For crowdsourcing experiments and research with human subjects, does the paper include the full text of instructions given to participants and screenshots, if applicable, as well as details about compensation (if any)? 
    \item[] Answer: \answerNA{}.
    \item[] Justification: The paper does not involve crowdsourcing, surveys, experiments with human participants, or human-subject data collection. There are therefore no participant instructions, screenshots, or compensation details to report.

\item {\bf Institutional review board (IRB) approvals or equivalent for research with human subjects}
    \item[] Question: Does the paper describe potential risks incurred by study participants, whether such risks were disclosed to the subjects, and whether Institutional Review Board (IRB) approvals (or an equivalent approval/review based on the requirements of your country or institution) were obtained?
    \item[] Answer: \answerNA{}.
    \item[] Justification: The paper does not involve human subjects or crowdsourced participants. IRB or equivalent human-subjects review is therefore not applicable.

\item {\bf Declaration of LLM usage}
    \item[] Question: Does the paper describe the usage of LLMs if it is an important, original, or non-standard component of the core methods in this research? Note that if the LLM is used only for writing, editing, or formatting purposes and does \emph{not} impact the core methodology, scientific rigor, or originality of the research, declaration is not required.
    \item[] Answer: \answerYes{}.
    \item[] Justification: AI tools were used to help check the manuscript's derivations and proofs, and to provide suggestions for improving and revising the derivations.

\end{enumerate}

\end{document}